\documentclass[11pt,onecolumn,draft]{IEEEtran}

\usepackage{epsf,psfrag,amssymb,amsfonts,cite,amsmath,yfonts}
\def\HOME{/include}
\newtheorem{theorem}{Theorem}
\newtheorem{example}{Example}
\newtheorem{corollary}{Corollary}
\newtheorem{lemma}{Lemma}
\newtheorem{definition}{Definition}

\def\psfancypar#1#2{\begingroup\def\par{\endgraf\endgroup\lineskiplimit=0pt}
               \setbox2=\hbox{\large\sc #2}
               \newdimen\tmpht \tmpht \ht2 \advance\tmpht by \baselineskip
               \font\hhuge=Times-Bold at \tmpht
               \setbox1=\hbox{{\hhuge #1}}
               \count7=\tmpht \count8=\ht1
               \divide\count8 by 1000 \divide\count7 by \count8 
               \tmpht=.001\tmpht\multiply\tmpht by \count7 
               \font\hhuge=Times-Bold at \tmpht
               \setbox1=\hbox{{\hhuge #1}}
               \noindent
                \hangindent1.05\wd1
               \hangafter=-2 {\hskip-\hangindent
               \lower1\ht1\hbox{\raise1.0\ht2\copy1}%
                \kern-0\wd1}\copy2\lineskiplimit=-1000pt}

\newcommand{\beq}{\begin{equation}}
\newcommand{\eeq}{\end{equation}}
\newcommand{\bqa}{\begin{eqnarray}}
\newcommand{\eqa}{\end{eqnarray}}
\newcommand{\bqn}{\begin{eqnarray*}}
\newcommand{\eqn}{\end{eqnarray*}}
\newcommand{\nn}{\nonumber}

\newcommand{\be}{\begin{enumerate}}
\newcommand{\ee}{\end{enumerate}}
\newcommand{\bi}{\begin{itemize}}
\newcommand{\ei}{\end{itemize}}
\newcommand{\bd}{\begin{description}}
\newcommand{\ed}{\end{description}}
\newcommand{\ba}{\begin{array}}
\newcommand{\ea}{\end{array}}
\newcommand{\bde}{\begin{definition}}
\newcommand{\ede}{\end{definition}}
\newcommand{\bex}{\begin{example}}
\newcommand{\eex}{\end{example}}

\def\boxit#1{\vbox{\hrule\hbox{\vrule\kern3pt
        \vbox{\kern3pt#1\kern3pt}\kern3pt\vrule}\hrule}}

\def\reals{ { {\rm  I \kern-0.15em R }  } }
\def\complex{ {\,{{\rm C} \kern-0.50em \raise0.20ex {  |}}\, }}

\def\Lambdabf{\mbox{$ \bf \Lambda $}}

\def\0bf{{\bf 0}}
\def\1bf{{\bf 1}}
\def\2bf{{\bf 2}}
\def\3bf{{\bf 3}}
\def\4bf{{\bf 4}}
\def\5bf{{\bf 5}}
\def\6bf{{\bf 6}}
\def\7bf{{\bf 7}}
\def\8bf{{\bf 8}}
\def\9bf{{\bf 9}}

\def\xbf{{\bf x}}

\def\xbf{{\bf x}}

\def\Ibf{{\bf I}}

\def\Kbf{{\bf K}}

\def\Qbf{{\bf Q}}
\def\Rbf{{\bf R}}
\def\Sbf{{\bf S}}

\def\Ubf{{\bf U}}

\def\Xbf{{\bf X}}

\def\Emat{\mathcal{E}}

\def\Rxx{\Rbf_{\ssstyle X\kern-.1em X}}

\let\ssstyle=\scriptscriptstyle

\def\Kout{\setbox1=\hbox{\Huge\bf K}\hbox to
1.05\wd1{\hspace{.05\wd1}
}}
\def\Sout{\setbox1=\hbox{\Huge\bf S}\hbox to 1.05\wd1{\hspace{.05\wd1}
}}

\allowdisplaybreaks[2]
\begin{document}
\title{A New Outer Bound and the Noisy-Interference Sum-Rate
Capacity for Gaussian Interference Channels}
\author{Xiaohu Shang, Gerhard Kramer, and Biao Chen\thanks{X. Shang and B. Chen are with
Syracuse University, Department of EECS, 335 Link Hall, Syracuse,
NY 13244. Phone: (315)443-3332. Email: xshang@syr.ed and
bichen@ecs.syr.edu. G. Kramer is with Bell Labs, Alcatel-Lucent,
600 Mountain Ave. Murray Hill, NJ 07974-0636 Phone: (908)582-3964.
Email: gkr@research.bell-labs.com. }} \maketitle

\begin{abstract}
A new outer bound on the capacity region of Gaussian interference
channels is developed. The bound combines and improves existing
genie-aided methods and is shown to give the sum-rate capacity for
{\em noisy interference} as defined in this paper. Specifically,
it is shown that if the channel coefficients and power constraints
satisfy a simple condition then single-user detection at each
receiver is sum-rate optimal, i.e., treating the interference as
noise incurs no loss in performance. This is the first concrete
(finite signal-to-noise ratio) capacity result for the Gaussian
interference channel with weak to moderate interference.
Furthermore, for certain mixed (weak and strong) interference
scenarios, the new outer bounds give a corner point of the
capacity region.

\end{abstract}
 \maketitle

\section{Introduction}

The interference channel (IC) models communication systems where
transmitters communicate with their respective receivers while
causing interference to all other receivers. For a two-user
Gaussian IC, the channel output can be written in the standard
form\cite{Carleial:78IT} \bqn
Y_1&=&X_1+\sqrt{a}X_2+Z_1,\\Y_2&=&\sqrt{b}X_1+X_2+Z_2,\eqn where
$\sqrt{a}$ and $\sqrt{b}$ are channel coefficients, $X_i$ and
$Y_i$ are the transmit and receive signals, and where the
user/channel input sequence $X_{i1},X_{i2},\cdots,X_{in}$ is
subject to the power constraint $\sum_{j=1}^n\Emat(X^2_{ij})\leq
nP_i$, $i=1,2$. The transmitted signals $X_1$ and $X_2$ are
statistically independent. The channel noises $Z_1$ and $Z_2$ are
possibly correlated unit variance Gaussian random variables, and
$(Z_1,Z_2)$ is statistically independent of $(X_1,X_2)$. In the
following, we denote this Gaussian IC as IC$(a,b,P_1,P_2)$.

The capacity region of an IC is defined as the closure of the set
of rate pairs $(R_1,R_2)$ for which both receivers can decode
their own messages with arbitrarily small positive error
probability. The capacity region of a Gaussian IC is known only
for three cases: \bi \item $a=0$, $b=0$. \item $a\geq 1$, $b\geq
1$: see \cite{Carleial:75IT,Sato:81IT,Han&Kobayashi:81IT}. \item
$a=0$, $b\geq 1$; or $a\geq 1$, $b=0$: see \cite{Costa:85IT} \ei
For the second case both receivers can decode the messages of both
transmitters. Thus this IC acts as two multiple access channels
(MACs), and the capacity region for the IC is the intersection of
the capacity region of the two MACs. However, when the
interference is weak or moderate, the capacity region is still
unknown. The best inner bound of the capacity region is obtained
in \cite{Han&Kobayashi:81IT} by using superposition coding and
joint decoding. A simplified form of the Han-Kobayashi region was
given by Chong-Motani-Garg \cite{Chong-etal:06IT_submission},
\cite{Kramer:06Zurich}. Various outer bounds have been developed
in
\cite{Sato:77IT,Carleial:83IT,Kramer:04IT,Etkin-etal:07IT_submission,Telatar&Tse:07ISIT}.
Sato's outer bound in \cite{Sato:77IT} is derived by allowing the
receivers to cooperate. Carleial's outer bound in
\cite{Carleial:83IT} is derived by decreasing the noise power.
Kramer in \cite{Kramer:04IT} presented two outer bounds. The first
is obtained by providing each receiver with just enough
information to decode both messages. The second outer bound is
obtained by reducing the IC to a degraded broadcast channel. Both
of these two bounds dominate the bounds by Sato and Carleial. The
recent outer bounds by Etkin, Wang, and Tse in
\cite{Etkin-etal:07IT_submission} are also based on genie-aided
methods, and they show that Han and Kobayashi's inner bound is
within one bit or a factor of two of the capacity region. This
result can also be established by the methods of Telatar and Tse
\cite{Telatar&Tse:07ISIT}. We remark that neither of the bounds of
\cite{Kramer:04IT} and \cite{Etkin-etal:07IT_submission} implies
each other. But as a rule of thumb, our numerical results show
that the bounds of \cite{Kramer:04IT} are better at low SNR while
those of \cite{Etkin-etal:07IT_submission} are better at high SNR.
The bounds of \cite{Telatar&Tse:07ISIT} are not amenable to
numerical evaluation since the optimal distributions of the
auxiliary random variables are unknown. None of the above outer
bounds is known to be tight for the general Gaussian IC.

In this paper, we present a new outer bound on the capacity region
of Gaussian ICs that improves on the bounds of
\cite{Kramer:04IT,Etkin-etal:07IT_submission}. The new bounds are
based on a genie-aided approach and a recently proposed extremal
inequality \cite{Liu&Viswanath:06IT}. Unlike the genie-aided
method used in \cite[Theorem 1]{Kramer:04IT}, neither receiver is
required to decode the messages from the other transmitter. Based
on this outer bound, we obtain new sum-rate capacity results
(Theorem \ref{thm:sumcapacity} and \ref{thm:sumcapacity_mixed})
for ICs satisfying some channel coefficient and power constraint
conditions. We show that the sum-rate capacity can be achieved by
treating the interference as noise when both the channel gain and
the power are weak. We say that such channels have {\em noisy
interference}. For this kind of noisy interference, the simple
single-user transmission and detection strategy is sum-rate
optimal. In Theorem \ref{thm:sumcapacity_mixed}, we show that for
ICs with $a>1, 0<b<1$ and satisfying another condition, the
sum-rate capacity is achieved by letting user $1$ fully recover
messages from user $2$ first before decoding its own message,
while user $2$ only recovers its own messages.

This paper is organized as follows. In Section II, we present a
new genie-aided outer bound and the resulting sum-rate capacity
for certain Gaussian ICs. We prove these results in Section III.
Numerical examples are given in Section IV, and Section V
concludes the paper.

\section{Main Results}

\subsection{General outer bound}
The following is a new outer bound on the capacity region of
Gaussian ICs.

\begin{theorem}
If the rates $(R_1,R_2)$ are achievable for IC$(a,b,P_1,P_2)$ with
$0<a<1,0<b<1$, they must satisfy the following constraints
(\ref{eq:constraint1})-(\ref{eq:constraint3}) for $\mu>0$,
$\frac{1+bP_1}{b+bP_1}\leq\eta_1\leq\frac{1}{b}$ and
$a\leq\eta_2\leq\frac{a+aP_2}{1+aP_2}$: \bqa R_1+\mu
R_2&\leq&\min_{\substack{\rho_i\in[0,1]\\\left(\sigma_1^2,\sigma_2^2\right)\in\Sigma}}\frac{1}{2}\log\left(1+\frac{P_1^*}{\sigma_1^2}\right)-\frac{1}{2}\log\left(aP_2^*+1-\rho_1^2\right)+\frac{1}{2}\log\left(1+P_1+aP_2-\frac{(P_1+\rho_1\sigma_1)^2}{P_1+\sigma_1^2}\right)\nn\\
&&\hspace{.2in}+\frac{\mu}{2}\log\left(1+\frac{P_2^*}{\sigma_2^2}\right)-\frac{\mu}{2}\log\left(bP_1^*+1-\rho_2^2\right)+\frac{\mu}{2}\log\left(1+P_2+bP_1-\frac{(P_2+\rho_2\sigma_2)^2}{P_2+\sigma_2^2}\right),\label{eq:constraint1}\\
R_1+\eta_1R_2&\leq&\frac{1}{2}\log\left(1+\frac{b\eta_1-1}{b-b\eta_1}\right)-\frac{\eta_1}{2}\log\left(1+\frac{b\eta_1-1}{1-\eta_1}\right)+\frac{\eta_1}{2}\log\left(1+bP_1+P_2\right),\label{eq:constraint2}\\
R_1+\eta_2R_2&\leq&\frac{1}{2}\log\left(1+P_1+aP_2\right)-\frac{1}{2}\log\left(1+\frac{a-\eta_2}{\eta_2-1}\right)+\frac{\eta_2}{2}\log\left(1+\frac{a-\eta_2}{a\eta_2-a}\right),\label{eq:constraint3}\eqa
where \bqa\Sigma&=&\left\{\begin{array}{ll}
\left\{\left(\sigma_1^2,\sigma_2^2\right)\quad\left.|\quad\sigma_1^2>0,\quad 0<\sigma_2^2\leq\frac{1-\rho_1^2}{a}\right.\right\},\quad\textrm{ if }\mu\geq 1, \\
\left\{\left(\sigma_1^2,\sigma_2^2\right)\quad\left.|\quad 0<\sigma_1^2\leq\frac{1-\rho_2^2}{b},\quad \sigma_2^2>0\right.\right\},\quad\textrm{ if }\mu<1, \\
 \end{array}\right.\eqa
and if $\mu\geq 1$ we have \bqa
P_1^*&=&\left\{\begin{array}{ll}P_1,&\quad 0<\sigma_1^2\leq\left(\left(\frac{1}{\mu}-1\right)P_1+\frac{1-\rho_2^2}{b\mu}\right)^+,\\
\frac{1-\rho_2^2-b\mu\sigma_1^2}{b\mu-b},&\quad
\left(\left(\frac{1}{\mu}-1\right)P_1+\frac{1-\rho_2^2}{b\mu}\right)^+<\sigma_1^2\leq\frac{1-\rho_2^2}{b\mu},\\
0,&\quad\sigma_1^2>\frac{1-\rho_2^2}{b\mu},\end{array}\right.\label{eq:P1muL}\\
P_2^*&=&P_2,\qquad\qquad\qquad
0<\sigma_2^2\leq\frac{1-\rho_1^2}{a}, \label{eq:P2muL}\eqa where
$(x)^+\triangleq\max\{x,0\}$, and if $0<\mu<1$ we have\bqa
P_1^*&=&P_1,\qquad\qquad\qquad\quad
0<\sigma_1^2\leq\frac{1-\rho_2^2}{b},\label{eq:P1muS}\\
P_2^*&=&\left\{\begin{array}{ll}P_2,&\quad 0<\sigma_2^2\leq\left(\left(\mu-1\right)P_2+\frac{\mu\left(1-\rho_1^2\right)}{a}\right)^+,\\
\frac{\mu\left(1-\rho_1^2\right)-a\sigma_2^2}{a-a\mu},&\quad
\left(\left(\mu-1\right)P_2+\frac{\mu\left(1-\rho_1^2\right)}{a}\right)^+<\sigma_2^2\leq\frac{\mu\left(1-\rho_1^2\right)}{a},\\
0,&\quad\sigma_2^2>\frac{\mu\left(1-\rho_1^2\right)}{a}.\end{array}\right.\label{eq:P2muS}
 \eqa\label{thm:region}
\end{theorem}

Remark 1: The bounds (\ref{eq:constraint1})-(\ref{eq:constraint3})
are obtained by providing different genie-aided signals to the
receivers. There is overlap of the range of $\mu$, $\eta_1$, and
$\eta_2$, and none of the bounds uniformly dominates the other two
bounds. Which one of them is active depends on the channel
conditions and the rate pair.

Remark 2: Equations (\ref{eq:constraint2}) and
(\ref{eq:constraint3}) are outer bounds for the capacity region of
a Z-IC, and a Z-IC is equivalent to a degraded IC
\cite{Costa:85IT}. For such channels, it can be shown that
(\ref{eq:constraint2}) and (\ref{eq:constraint3}) are the same as
the outer bounds in \cite{Sato:78IT}. For
$0\leq\eta_1\leq\frac{1+bP_1}{b+bP_1}$ and
$\eta_2\geq\frac{a+aP_2}{1+aP_2}$, the bounds in
(\ref{eq:constraint2}) and (\ref{eq:constraint3}) are tight for a
Z-IC (or degraded IC) because there is no power sharing between
the transmitters. Consequently, $\frac{1+bP_1}{b+bP_1}$ and
$\frac{a+aP_2}{1+aP_2}$ are the negative slopes of the tangent
lines for the capacity region at the corner points.

Remark 3: The bounds in
(\ref{eq:constraint2})-(\ref{eq:constraint3}) turn out to be the
same as the bounds in \cite[Theorem 2]{Kramer:04IT}. We show this
by proving that (\ref{eq:constraint3}) is equivalent to \cite[page
584, (37)-(38)]{Kramer:04IT} but with equalities rather than
inequalities. Consider the rates\bqa
R_1&=&\frac{1}{2}\log\left(1+\frac{P_1^\prime}{P_2^\prime+1/a}\right)\label{eq:Kramer1}\\
R_2&=&\frac{1}{2}\log\left(1+P_2^\prime\right)\label{eq:Kramer2}\\
P_1^\prime+P_2^\prime&=&\frac{P_1}{a}+P_2\eqa for $0\leq
P_1^\prime\leq P_1$. We rewrite (\ref{eq:Kramer1}) and
(\ref{eq:Kramer2}) in the form of the weighted sum \bqa R_1+\alpha
R_2=\frac{1}{2}\log\left(1+\frac{P_1^\prime}{P_2^\prime+1/a}\right)+\frac{\alpha}{2}\log\left(1+P_2^\prime\right).\label{eq:KramerWeighted}\eqa
Observe that (\ref{eq:KramerWeighted}) represents a line with
slope $\alpha$ where \bqa \alpha&=&-\frac{\partial
R_1}{\partial R_2}\nn\\
&=&-\frac{\partial R_1}{\partial P_2^\prime}\left/\frac{\partial
R_2}{\partial P_2^\prime}\right.\nn\\
&=&-\frac{\partial
\log\left(1+\frac{P_1/a+P_2-P_2^\prime}{P_2^\prime+1/a}\right)}{\partial
P_2^\prime}\left/\frac{\partial\log\left(1+P_2^\prime\right)}{\partial
P_2^\prime}\right.\nn\\
&=&\frac{a+aP_2^\prime}{1+aP_2^\prime}.\label{eq:weight}\eqa We
thus obtain \bqa
P_2^\prime=\frac{a-\alpha}{a\alpha-a}.\label{eq:P1prime} \eqa
Substituting (\ref{eq:P1prime}) into (\ref{eq:KramerWeighted}), we
have \bqn R_1+\alpha
R_2=\frac{1}{2}\log\left(1+P_1+aP_2\right)-\frac{1}{2}\log\left(1+\frac{a-\alpha}{\alpha-1}\right)+\frac{\alpha}{2}\log\left(1+\frac{a-\alpha}{a\alpha-a}\right),\eqn
which is the same as (\ref{eq:constraint3}). The relation
$a\leq\alpha\leq\frac{a+aP_2}{1+aP_2}$ follows from
(\ref{eq:weight}) and $0\leq P_2^\prime\leq P_2$.

Remark 4: The bounds in \cite[Theorem 2]{Kramer:04IT} are obtained
by getting rid of one of the interference links to reduce the IC
into a Z interference channel (or Z-IC, see \cite{Costa:85IT}).
Next, the proof in \cite{Kramer:04IT} allowed the transmitters to
share their power, which further reduces the Z-IC into a degraded
broadcast channel. Then the capacity region of this degraded
broadcast channel is an outer bound for the capacity region of the
original IC. The bounds in (\ref{eq:constraint2}) and
(\ref{eq:constraint3}) are also obtained by reducing the IC to a
Z-IC. Although we do not explicitly allow the transmitters to
share their power, it is interesting that these bounds are
equivalent to the bounds in \cite[Theorem 2]{Kramer:04IT} with
power sharing. In fact, a careful examination of our new
derivation reveals that power sharing is implicitly assumed. For
example, for the term
$h\left(X_1^n+Z_1^n\right)-\eta_1h\left(\sqrt{b}X_1^n+Z_2^n\right)$
of (\ref{eq:BCbound}) below, user $1$ uses power
$P_1^*=\frac{b\eta_1-1}{b-b\eta_1}\leq P_1$ , while for the term
$\eta_1h\left(Y_2^n\right)$ user $1$ uses all the power $P_1$.
This is equivalent to letting user $1$ use the power $P_1^*$ for
both terms, and letting user $2$ use a power that exceeds $P_2$.
To see this, consider (\ref{eq:BCbound}) below and write \bqn
n(R_1+\eta_1
R_2)&\leq&\frac{n}{2}\log\left(P_1^*+1\right)-\frac{n\eta_1}{2}\log\left(bP_1^*+1\right)+\frac{n\eta_1}{2}\log\left(1+bP_1+P_2\right)+n\epsilon\\
&=&\frac{n}{2}\log\left(P_1^*+1\right)-\frac{n\eta_1}{2}\log\left(bP_1^*+1\right)+\frac{n\eta_1}{2}\log\left(1+bP_1^*+P_2+b(P_1-P_1^*)\right)+n\epsilon\\
&=&\frac{n}{2}\log\left(P_1^\prime+1\right)-\frac{n\eta_1}{2}\log\left(bP_1^\prime+1\right)+\frac{n\eta_1}{2}\log\left(1+bP_1^\prime+P_2^\prime\right)+n\epsilon,
 \eqn
where $P_1^\prime\triangleq P_1^*$, and $P_2^\prime\triangleq
P_2+b(P_1-P_1^*)$. Therefore, one can assume that user $2$ uses
extra power provided by user $1$.

Remark 5: Theorem \ref{thm:region} improves \cite[Theorem
3]{Etkin-etal:07IT_submission}. Specifically, for the three
sum-rate bounds of \cite[Theorem 3]{Etkin-etal:07IT_submission},
the first bound can be obtained from (\ref{eq:BCbound}) with
$P_1^*=P_1$ in (\ref{eq:starP1}). Therefore, the bound in
(\ref{eq:constraint2}) is tighter than the first sum-rate bound of
\cite[Theorem 3]{Etkin-etal:07IT_submission}. Similarly, the bound
in (\ref{eq:constraint3}) is tighter than the second sum-rate
bound of \cite[Theorem 3]{Etkin-etal:07IT_submission}. The third
sum-rate bound in \cite[Theorem 3]{Etkin-etal:07IT_submission} is
a special case of (\ref{eq:constraint1}) with
$\sigma_1^2=\frac{1}{b},\sigma_2^2=\frac{1}{a},\rho_1=\rho_2=0$.

Remark 6: Our outer bound is not always tighter than that of
\cite{Etkin-etal:07IT_submission} for all rate points. The reason
is that in \cite[last two equations of
(39)]{Etkin-etal:07IT_submission}, different genie-aided signals
are provided to the same receiver. Our outer bound can also be
improved in a similar and more general way by providing different
genie-aided signals to the receivers. Specifically the starting
point of the bound is \bqa n\left(R_1+\mu
R_2\right)\leq\sum_{i=1}^k\lambda_iI\left(X_1^n;Y_1^n,U_i\right)+\sum_{j=1}^m\mu_iI\left(X_2^n;Y_2^n,W_j\right)+n\epsilon,\label{eq:extension}\eqa
where
$\sum_{i=1}^k\lambda_i=1,\sum_{j=1}^m\mu_j=\mu,\lambda_i>0,\mu_j>0$.

\subsection{Sum-rate capacity for noisy interference}

The outer bound in Theorem \ref{thm:region} is in the form of an
optimization problem. Four parameters
$\rho_1,\rho_2,\sigma_1^2,\sigma_2^2$ need to be optimized for
different choices of the weights $\mu,\eta_1,\eta_2$. When
$\mu=1$, Theorem \ref{thm:region} leads directly to the following
sum-rate capacity result.

\begin{theorem}
For the IC$(a,b,P_1,P_2)$ satisfying \bqa
\sqrt{a}(bP_1+1)+\sqrt{b}(aP_2+1)\leq 1,\label{eq:power}\eqa
 the sum-rate capacity is \bqa
C=\frac{1}{2}\log\left(1+\frac{P_1}{1+aP_2}\right)+\frac{1}{2}\log\left(1+\frac{P_2}{1+bP_1}\right).\label{eq:sumcapacity}\eqa
\label{thm:sumcapacity}\end{theorem}

Remark 7: The sum-rate capacity for a Z-IC with $a=0$, $0<b<1$ is
a special case of Theorem \ref{thm:sumcapacity} since
(\ref{eq:power}) is satisfied. The sum capacity is therefore given
by (\ref{eq:sumcapacity}).

Theorem \ref{thm:sumcapacity} follows directly from Theorem
\ref{thm:region} with $\mu=1$. It is remarkable that a genie-aided
bound is tight if (\ref{eq:power}) is satisfied since the genie
provides extra signals to the receivers without increasing the
rates. This situation is reminiscent of the recent capacity
results for vector Gaussian broadcast channels (see
\cite{Weingarten-etal:06IT}). Furthermore, the sum-rate capacity
(\ref{eq:sumcapacity}) is achieved by treating the interference as
noise. We therefore refer to channels satisfying (\ref{eq:power})
as ICs with {\em noisy interference}. Note that (\ref{eq:power})
involves both channel gains $a,b$ and both powers $P_1$ and $P_2$.
The constraint (\ref{eq:power}) implies that
\bqa\sqrt{a}+\sqrt{b}\leq 1.\label{eq:channelgain}\eqa Moreover,
as shown in Fig. \ref{fig:weakplot}, the powers $P_1$ and $P_2$
must be inside the triangle defined by: \bqa P_1&\geq& 0,\nn\\
P_2&\geq& 0,\nn\\b\sqrt{a}P_1+a\sqrt{b}P_2&\leq&
1-\sqrt{a}-\sqrt{b}.\label{eq:weakpower}\eqa These constraints can
be considered as a counterpart of the IC with very strong
interference \cite{Carleial:75IT} whose powers should be inside the rectangle defined in Fig. \ref{fig:strongplot}: \bqn &&a>1,b>1,\\
&&0\leq P_1\leq a-1,\\&& 0\leq P_2\leq b-1.\eqn

The ICs with noisy interference and ICs with very strong
interference are two extreme cases in terms of the decoding
strategy to achieve the sum-rate capacity. In the former case, the
sum-rate capacity is achieved by treating interference as noise,
while in the latter case, the interference is decoded before, or
together with, the intended messages.

 For symmetric Gaussian ICs
with $a=b$ and $P_1=P_2$, the conditions in (\ref{eq:channelgain})
and  (\ref{eq:weakpower})
become \bqa a=b&\leq&\frac{1}{4},\\
P_1=P_2=P&\leq&\frac{\sqrt{a}-2a}{2a^2}. \label{eq:symm_weak}\eqa
``Noisy interference" is therefore ``weaker" than ``weak
interference" as defined in \cite{Costa:85IT} and
\cite{Sason:04IT}, namely $a\leq\frac{\sqrt{1+2P}-1}{2P}$ or \bqa
P\leq\frac{1-2a}{a^2}.\label{eq:weak}\eqa Recall that
\cite{Sason:04IT} showed that for ``weak interference" satisfying
(\ref{eq:weak}), treating interference as noise achieves larger
sum rate than time-or frequency-division multiplexing (TDM/FDM),
and \cite{Costa:85IT} claimed that in ``weak interference" the
largest known achievable sum rate is achieved by treating the
interference as noise.

{\begin{figure}[htp] \centerline{\leavevmode \epsfxsize=5in
\epsfysize=3.4in \epsfbox{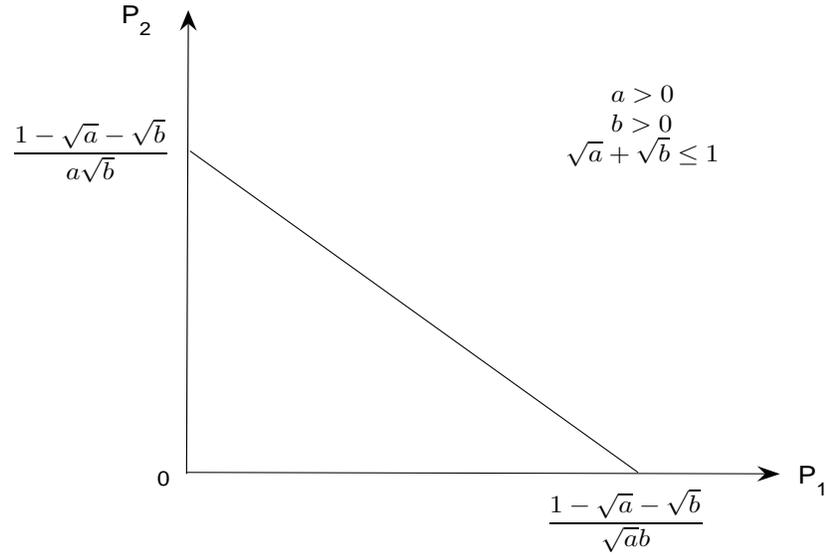}}\caption{Power region for
the IC with noisy interference.} \label{fig:weakplot}\end{figure}

\begin{figure}[htp]
\centerline{\leavevmode \epsfxsize=5in \epsfysize=3.4in
\epsfbox{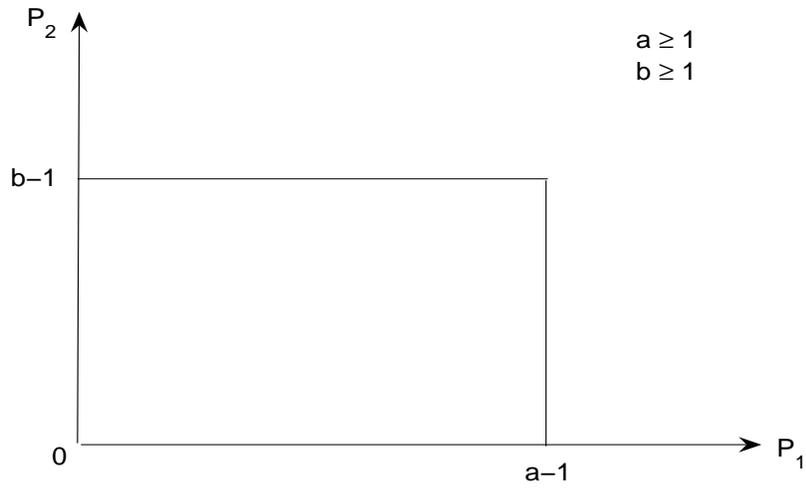}}\caption{Power region for the IC with
very strong interference.} \label{fig:strongplot}\end{figure} }

\subsection{Capacity region corner point}

\begin{theorem}
For an IC$(a,b,P_1,P_2)$ with $a>1$, $0<b<1$, the sum-rate
capacity is \bqa
C=\frac{1}{2}\log\left(1+P_1\right)+\frac{1}{2}\log\left(1+\frac{P_2}{1+bP_1}\right)\label{eq:sumcapacity_mixed}\eqa
when the following condition holds \bqa (1-ab)P_1\leq
a-1.\label{eq:constraint_mixed}\eqa A similar result follows by
swapping $a$ and $b$, and $P_1$ and $P_2$.
\label{thm:sumcapacity_mixed}\end{theorem}

Under the constraint (\ref{eq:constraint_mixed}), we have the
following inequality: \bqa
\frac{1}{2}\log\left(1+\frac{P_2}{1+bP_1}\right)\leq\frac{1}{2}\log\left(1+\frac{aP_2}{1+P_1}\right)\label{eq:icmixed}.\eqa
Therefore, the sum-rate capacity is achieved by a simple scheme:
user $1$ transmits at the maximum rate and user $2$ transmits at
the rate that both receivers can decode its message with
single-user detection. Observe further that this rate pair permits
$R_2=\frac{1}{2}\log\left(1+\frac{aP_2}{1+P_1}\right)$ when $R_1$
reaches its maximum. Such a rate constraint was considered in
\cite[Theorem 1]{Costa:85IT} which established a corner point of
the capacity region. However it was pointed out in
\cite{Sason:04IT} that the proof in \cite{Costa:85IT} was flawed.
Theorem \ref{thm:sumcapacity_mixed} shows that the rate pair of
\cite{Sason:04IT} is in fact a corner point of the capacity region
when $a>1, 0<b<1$ and (\ref{eq:constraint_mixed}) is satisfied,
and this rate pair achieves the sum-rate capacity.

 The sum-rate
capacity of the degraded IC $(ab=1,0<b<1)$ is a special case of
Theorem \ref{thm:sumcapacity_mixed}. Besides this example, there
are two other kinds of ICs to which Theorem
\ref{thm:sumcapacity_mixed} applies. The first case is $ab>1$. In
this case, $P_1$ can be any positive value. The second case is
$ab<1$ and $P_1\leq\frac{a-1}{1-ab}$. For both cases, the signals
from user $2$ can be decoded first at both receivers.

\subsection{State of the Art}
We reiterate that both Theorems \ref{thm:sumcapacity} and
\ref{thm:sumcapacity_mixed} are direct results of Theorem
\ref{thm:region}, and Theorem \ref{thm:region} is derived by
having a genie provide extra information to the receivers. We
summarize the sum-rate capacity for Gaussian ICs from Theorems
\ref{thm:sumcapacity} and \ref{thm:sumcapacity_mixed} and previous
results in \cite{Carleial:75IT,Han&Kobayashi:81IT,Sato:81IT}. In
Fig. \ref{fig:ab}, four curves $ab=1$, $a=1$, $b=1$, and
$\sqrt{a}+\sqrt{b}\leq 1$ divide channel gain plane into $7$
regimes. The sum-rate capacity for each regime under certain power
constraints is shown in Tab. \ref{tab:ab}.

\section{Proofs of the Main Results}

We introduce some notation. We write vectors and matrices by using
a bold font (e.g., $\Xbf$ and $\Sbf$). When useful we also write
vectors with length $n$ using the notation $X^n$. The $i$th entry
of the vector $\Xbf$ (or $X^n$) is denoted as $X_i$. Random
variables are written as uppercase letters (e.g., $X_i$) and their
realizations as the corresponding lowercase letter (e.g., $x_i$).
We usually write probability densities and distributions as
$p(\xbf)$ if the argument of $p(\cdot)$ is a lowercase version of
the random variable corresponding to this density or distribution.
The notation $h(\Xbf)$ and $\textrm{Cov}(\Xbf)$ refers to the
respective differential entropy and covariance matrix of $\Xbf$.
The notation of $U|V=v$ and $U|V$ denotes the random variable $U$
conditioned on the event $V=v$ and the random variable $V$,
respectively.

\begin{figure}[htp]
\centerline{\leavevmode \epsfxsize=5in \epsfysize=4.5in
\epsfbox{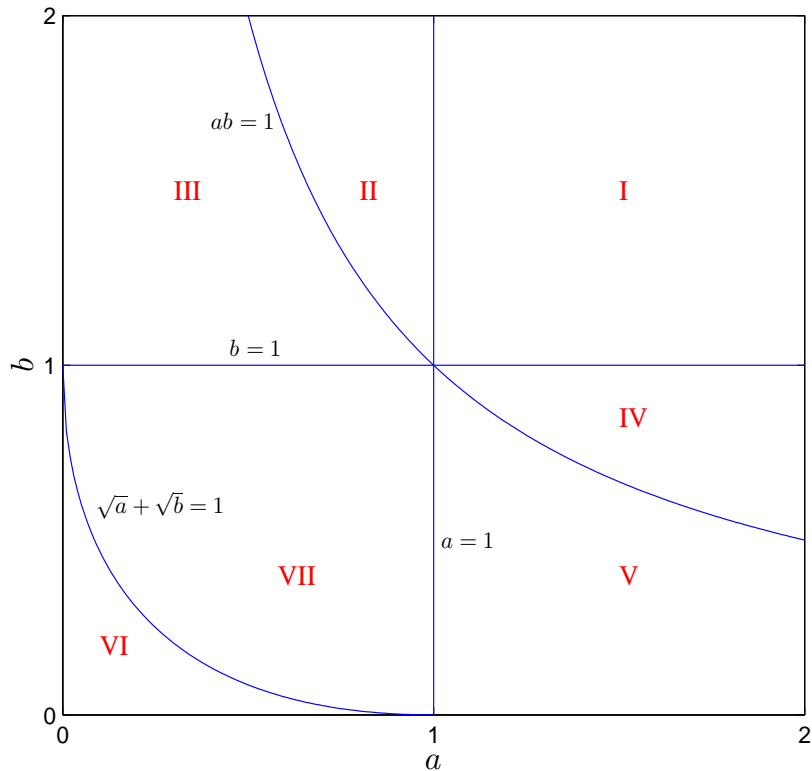}}\caption{Gaussian IC channel coefficient regimes
for Tab. \ref{tab:ab}} \label{fig:ab}\end{figure}

\begin{table}[htp]\caption{\label{tab:ab} Sum-rate capacity.}
{\small \centerline{\begin{tabular}{|c|c|c|c|}
  \hline
   &$(a,b)$ & $(P_1,P_2)$ & sum-rate capacity \\
  \hline
   I & $a\geq 1, b\geq 1$& $P_1>0,P_2>0$ & $\min\left\{\begin{array}{c}
    \frac{1}{2}\log(1+P_1)+\frac{1}{2}\log(1+P_2) \\
    \frac{1}{2}\log(1+P_1+aP_2) \\
    \frac{1}{2}\log(1+bP_1+P_2) \\
  \end{array}\right\}$ \\
   \hline
II & $ab\geq 1,a\leq 1$& $P_1>0,P_2>0$ &
$\frac{1}{2}\log\left(1+\frac{P_1}{1+aP_2}\right)+\frac{1}{2}\log(1+P_2)$\\
  \hline
 III & $ab\leq 1,b\geq 1$& $P_1>0,P_2\leq\frac{b-1}{1-ab}$ & \textrm{same as above}\\
  \hline
 IV & $ab\geq 1,b\leq 1$& $P_1>0,P_2>0$ &
$\frac{1}{2}\log(1+P_1)+\frac{1}{2}\log\left(1+\frac{P_2}{1+bP_1}\right)$\\
  \hline
   V & $ab\leq 1,a\geq 1$& $P_1\leq\frac{a-1}{1-ab},P_2>0$ &\textrm{same as above}\\
  \hline
     VI & $\sqrt{a}+\sqrt{b}\leq 1$ & $\sqrt{a}(1+bP_1)+\sqrt{b}(1+aP_2)\leq 1$ & $\frac{1}{2}\log\left(1+\frac{P_1}{1+aP_2}\right)+\frac{1}{2}\log\left(1+\frac{P_2}{1+bP_1}\right)$\\
  \hline
   VII & $\sqrt{a}+\sqrt{b}>1,a<1,b<1$& $P_1>0,P_2>0$ &\textrm{unknown}\\
  \hline
\end{tabular}}}
\end{table}

The proof utilizes the extremal inequalities introduced in
\cite{Liu&Viswanath:06IT}. We present them below for completeness.

\begin{lemma}
\cite[Theorem 1]{Liu&Viswanath:06IT} For any $\mu\geq 1$ and any
positive semi-definite $\Sbf$, a Gaussian $\Xbf$ is an optimal
solution of the following optimization problem:

\bqn \max_{p(\xbf)}&&h\left(\Xbf+\Ubf_1\right)-\mu
h\left(\Xbf+\Ubf_2\right)\\
\textrm{subject to}&&\textrm{Cov}(\Xbf)\preceq\Sbf,\eqn where
$\Ubf_1$ and $\Ubf_2$ are Gaussian vectors with strictly positive
definite covariance matrices $\Kbf_1$ and $\Kbf_2$, respectively,
and the maximization is over all $\Xbf$ independent of $\Ubf_1$
and $\Ubf_2$. \label{lemma:lemma1}
\end{lemma}

\begin{lemma}
\cite[Corollary 4]{Liu&Viswanath:06IT} For any real number $\mu$
and any positive semi-definite $\Sbf$, a Gaussian $\Xbf$ is an
optimal solution of the following optimization problem: \bqn
\max_{p(\xbf)}&&h\left(\Xbf+\Ubf_1\right)-\mu
h\left(\Xbf+\Ubf_1+\Ubf\right)\\
\textrm{subject to}&&\textrm{Cov}(\Xbf)\preceq\Sbf,\eqn where
$\Ubf_1$ and $\Ubf$ are two independent Gaussian vectors with
strictly positive definite covariance matrices $\Kbf_1$ and
$\Kbf$, respectively, and the maximization is over all $\Xbf$
independent of $\Ubf_1$ and $\Ubf_2$. \label{lemma:lemma2}
\end{lemma}

For example, consider the following optimization problem
 \bqa
\max_{p\left(\xbf\right)}&&h\left(\Xbf+\Ubf_1\right)-\mu
h\left(\Xbf+\Ubf_2\right)\nn\\
\textrm{subject to}&&\frac{1}{n}\textrm{tr}\left(\Sbf\right)\leq
P,\quad
\Sbf=\Emat\left(\Xbf\Xbf^{T}\right),\label{eq:optLemma}\eqa and
suppose that $\Sbf^*$ is the optimal covariance matrix for $\Xbf$.
When $\mu\geq 1$, the problem (\ref{eq:optLemma}) is equivalent to
the problem of Lemma \ref{lemma:lemma1} with $\Sbf$ replaced by
$\Sbf^*$. Similarly, when $\mu<1$ the problem (\ref{eq:optLemma})
is equivalent to the problem of Lemma \ref{lemma:lemma2} with
$\Sbf$ replaced by $\Sbf^*$ and $\Ubf_2=\Ubf_1+\Ubf$. Therefore a
Gaussian $\Xbf$ is optimal for problem (\ref{eq:optLemma}) in both
cases. We further have the following two simple optimization
results.

\begin{corollary}
The optimization problem of Lemma \ref{lemma:lemma1} with the
matrix constraint replaced by the trace constraint (or the problem
(\ref{eq:optLemma}) with $\mu\geq 1$) for the special case
$\textrm{Cov}(\Ubf_i)=\sigma_i^2\Ibf$, $i=1,2$, has the solution
$\textrm{Cov}(\Xbf)=P^*\Ibf$, where \bqa
P^*=\left\{\begin{array}{ll}
  0, &\quad 0<\sigma_2^2<\mu\sigma_1^2 \\
  \frac{\sigma_2^2-\mu\sigma_1^2}{\mu-1}, &\quad \mu\sigma_1^2\leq\sigma_2^2<\mu\sigma_1^2+(\mu-1)P \\
  P, &\quad\sigma_2^2\geq\mu\sigma_1^2+(\mu-1)P\\
\end{array}\right.\label{eq:optP_mu1}\eqa Alternatively, we can
write (\ref{eq:optP_mu1}) as \bqa P^*=\left\{\begin{array}{ll}
  P, &\quad 0<\sigma_1^2\leq\left(\frac{\sigma_2^2}{\mu}-\frac{\mu-1}{\mu}P\right)^+ \\
  \frac{\sigma_2^2-\mu\sigma_1^2}{\mu-1}, & \quad\left(\frac{\sigma_2^2}{\mu}-\frac{\mu-1}{\mu}P\right)^+<\sigma_1^2\leq\frac{\sigma_2^2}{\mu} \\
  0, &\quad \sigma_1^2>\frac{\sigma_2^2}{\mu} \\
\end{array}\right.\label{eq:opt_mu1P_eq}\eqa\label{corollary:c1}
\end{corollary}

\begin{corollary}
The optimization problem of Lemma \ref{lemma:lemma2} with the
matrix constraint replaced by the trace constraint (or the problem
of (\ref{eq:optLemma}) with $\mu<1$ and
$\sigma_1^2\leq\sigma_2^2$) for the special case
$\textrm{Cov}(\Ubf_1)=\sigma_1^2\Ibf$,
$\textrm{Cov}(\Ubf)=\left(\sigma_2^2-\sigma_1^2\right)\Ibf$, where
$\sigma_1^2\leq\sigma_2^2$, has the solution
$\textrm{Cov}(\Xbf)=P^*\Ibf$ where \bqa
P^*=P.\label{eq:optP_mu2}\eqa\label{corollary:c2}
\end{corollary}
\vspace{-.2in}
\begin{proof}Suppose the eigenvalue
decomposition of $\Sbf$ is $\Sbf=\Qbf\Lambda\Qbf^T$ and
$\Lambdabf=\textrm{diag}(\lambda_1,\dots,\lambda_n)$. Since
Gaussian $\Xbf$ is optimal, we have \bqn
&&h\left(\Xbf+\Ubf_1\right)-\mu
h\left(\Xbf+\Ubf_2\right)\\
&&=\frac{1}{2}\log\left[(2\pi e)^n\left|\Sbf+\sigma_1^2\Ibf\right|\right]-\frac{\mu}{2}\log\left[(2\pi e)^n\left|\Sbf+\sigma_2^2\Ibf\right|\right]\\
&&=\frac{1}{2}\log\left|\Lambdabf+\sigma_1^2\Ibf\right|-\frac{\mu}{2}\log\left|\Lambdabf+\sigma_2^2\Ibf\right|+\frac{1-\mu}{2}\log(2\pi e)^n\\
&&=\frac{1}{2}\sum_{i=1}^n\log\left(\lambda_i+\sigma_1^2\right)-\frac{\mu}{2}\sum_{i=1}^n\log\left(\lambda_i+\sigma_2^2\right)+\frac{1-\mu}{2}\log(2\pi e)^n\\
&&\triangleq f(\Lambdabf)\eqn By using the Lagrangian of
$f(\Lambdabf)$ with the constraint $\sum_{i=1}^n\lambda_i=nP$, it
can be shown that the optimal $\lambda_i$ is $\lambda_i^*=P^*$
with $P^*$ defined in (\ref{eq:optP_mu1})-(\ref{eq:optP_mu2}).
\end{proof}

 Finally we need another lemma to prove our main results.
\begin{lemma}
Suppose that $(U,V)$ is Gaussian with covariance matrix $\left[%
\begin{array}{cc}
  \sigma_1^2 & \rho\sigma_1\sigma_2 \\
  \rho\sigma_1\sigma_2 & \sigma_2^2 \\
\end{array}%
\right]$, $\sigma_1>0$, $\sigma_2>0$, $|\rho|<1$, and $W$ is
Gaussian with variance $\left(1-\rho^2\right)\sigma_1^2$. If the
discrete or continuous random variable $X$ is independent of
$(U,V)$ and $X$ is independent of $W$, then we have  \bqa
h\left(X+U|V\right)=h\left(X+W\right)\eqa
\label{lemma:lemma4}\end{lemma}
\begin{proof} We have\bqn h\left(X+U|V\right)&=&\int
f_V\left(v\right)h\left(X+U\left|V=v\right.\right)dv\\
&\stackrel{(a)}=&\int
f_V\left(v\right)h\left(\left.X+W^\prime+\frac{\rho\sigma_1}{\sigma_2}
V\right|V=v\right)dv\\
&\stackrel{(b)}=&\int f_V\left(v\right)h\left(X+W^\prime\right)dv\\
&=&h\left(X+W^\prime\right)\\
&=&h\left(X+W\right)\eqn where $W^\prime$ is identically
distributed as $W$ but independent of $(U,V)$. $(a)$ follows
because $\left(W^\prime+\frac{\rho\sigma_1}{\sigma_2}V,V\right)$
has the same joint distribution as $(U,V)$, $(b)$ follows because
$\frac{\rho\sigma_1}{\sigma_2}V$ becomes a constant when
conditioned on $V=v$.
\end{proof}

Since $U|V=v$ is also Gaussian distributed with mean value
$\frac{\rho\sigma_1}{\sigma_2}v$ and variance
$\left(1-\rho^2\right)\sigma_1^2$, Lemma \ref{lemma:lemma4} shows
that $U|V$ can be replaced by an equivalent Gaussian random
variable with the same variance.

\subsection{Proof of Theorem \ref{thm:region}}
Let $N_1$ and $N_2$ be two zero-mean Gaussian variables with
variances $\sigma_1^2$ and $\sigma_2^2$ respectively, and set
$\Emat(N_1Z_1)=\rho_1\sigma_1$ and $\Emat(N_2Z_2)=\rho_2\sigma_2$.
We further define $N_1^n$ and $N_2^n$ to be Gaussian vectors with
$n$ independent and identically distributed (i.i.d.) elements
distributed as $N_1$ and $N_2$, respectively.

Starting from Fano's inequality, we have that reliable
communication requires \bqa &&n(R_1+\mu
R_2)\nn\\
&&\leq I\left(X_1^n;Y_1^n\right)+\mu
I\left(X_2^n;Y_2^n\right)+n\epsilon\nn\\
&&\leq I\left(X_1^n;Y_1^n,X_1^n+N_1^n\right)+\mu
I\left(X_2^n;Y_2^n,X_2^n+N_2^n\right)+n\epsilon\nn\\
&&=I\left(X_1^n;X_1^n+N_1^n\right)+I\left(X_1^n;Y_1^n|X_1^n+N_1^n\right)+\mu
I\left(X_2^n;X_2^n+N_2^n\right)+\mu
I\left(X_2^n;Y_2^n|X_2^n+N_2^n\right)+n\epsilon\nn\\
&&=h\left(X_1^n+N_1^n\right)-h\left(N_1^n\right)+h\left(Y_1^n|X_1^n+N_1^n\right)-h\left(\sqrt{a}X_2^n+Z_1^n|N_1^n\right)+\mu
h\left(X_2^n+Z_2^n\right)-\mu h\left(N_2^n\right)\nn\\
&&\hspace{.1in}+\mu h\left(Y_2^n|X_2^n+N_2^n\right)-\mu
h\left(\sqrt{b}X_1^n+Z_2^n|N_2^n\right)+n\epsilon\label{eq:exprate}\eqa
where $\epsilon\rightarrow 0$ as $n\rightarrow\infty$. For
$h\left(Y_1^n|X_1^n+N_1^n\right)$, zero-mean Gaussian $X_1^n$ and
$X_2^n$ are optimal, and we have \bqa
\frac{1}{n}h\left(Y_1^n|X_1^n+N_1^n\right)&\leq&\frac{1}{n}\sum_{i=1}^nh\left(Y_{1i}|X_{1i}+N_1\right)\nn\\
&=&\frac{1}{n}\sum_{i=1}^n\left(h\left(X_{1i}+\sqrt{a}X_{2i}+Z_1,X_{1i}+N_1\right)-h\left(X_{1i}+N_1\right)\right)\nn\\
&=&\frac{1}{2n}\sum_{i=1}^n\log\left[2\pi
e\left(1+aP_{2i}+P_{1i}-\frac{\left(P_{1i}+\rho_1\sigma_1\right)^2}{P_{1i}+\sigma_1^2}\right)\right]\label{eq:eq1}\eqa
where $P_{1i}=\Emat(X_{1i}^2)$ and $P_{2i}=\Emat(X_{2i}^2)$.
Consider the function
\begin{align}
f(p_1,p_2) =
1+ap_2+p_1-\frac{(p_1+\rho_1\sigma_1)^2}{p_1+\sigma_1^2}
\end{align}
for which we compute \bqa &&\frac{\partial f}{\partial p_1}=
\frac{\sigma_1^2(\sigma_1-\rho_1)^2}{(p_1+\sigma_1^2)^2} \ge 0\label{eq:1order} \\
&&\frac{\partial f}{\partial p_2}=1 \\
&&\frac{\partial^2 f}{\partial p_1^2}=
-\frac{2\sigma_1^2(\sigma_1-\rho_1)^2}{(p_1+\sigma_1^2)^3} \le 0 \label{eq:2order}\\
&&\frac{\partial^2 f}{\partial p_2^2}=0 \\
 &&\frac{\partial^2
f}{\partial p_1 \partial p_1}=0. \eqa Since $\log(x)$ is concave
in $x$ we have that the logarithm in (\ref{eq:eq1}) is
concave in $(P_{1i},P_{2i})$. We thus have \bqa \frac{1}{n}h\left(Y_1^n|X_1^n+N_1^n\right)&\leq&\frac{1}{2}\log\left[2\pi e\left(1+\frac{a}{n}\sum_{i=1}^nP_{2i}+\frac{1}{n}\sum_{i=1}^nP_{1i}-\frac{\left(\frac{1}{n}\sum_{i=1}^nP_{1i}+\rho_1\sigma_1\right)^2}{\frac{1}{n}\sum_{i=1}^nP_{1i}+\sigma_1^2}\right)\right]\nn\\
&\leq&\frac{1}{2}\log\left[2\pi
e\left(1+aP_{2}+P_{1}-\frac{\left(P_{1}+\rho_1\sigma_1\right)^2}{P_{1}+\sigma_1^2}\right)\right]\label{eq:conditional1}\eqa
where the first inequality follows from Jensen's inequality, and
the second inequality follows from the block power constraints
$\frac{1}{n}\sum_{j=1}^nP_{ij}\leq P_i$, $i=1,2$, and
(\ref{eq:1order}).

For the same reason, we have\bqa
\frac{1}{n}h\left(Y_2^n|X_2^n+N_2^n\right)\leq\frac{1}{2}\log\left[2\pi
e\left(1+bP_{1}+P_{2}-\frac{\left(P_{2}+\rho_2\sigma_2\right)^2}{P_{2}+\sigma_2^2}\right)\right].\label{eq:conditional2}\eqa
Let $W_2^\prime=Z_2|N_2$, then $W_2^\prime$ is Gaussian
distributed with variance $1-\rho_2^2$. Define a new Gaussian
variable $W_2$ with variance $1-\rho_2^2$. From Lemma
\ref{lemma:lemma4} and Corollaries \ref{corollary:c1} and
\ref{corollary:c2} we have \bqa &&h\left(X_1^n+N_1^n\right)-\mu
h\left(\sqrt{b}X_1^n+Z_2^n|N_2^n\right)\nn\\
&&=h\left(X_1^n+N_1^n\right)-\mu
h\left(\sqrt{b}X_1^n+W_2^n\right)\nn\\
&&=h\left(X_1^n+N_1^n\right)-\mu
h\left(X_1^n+\frac{W_2^n}{\sqrt{b}}\right)-\frac{n\mu}{2}\log b\nn\\
&&\leq\frac{n}{2}\log\left[2\pi
e\left(P_1^*+\sigma_1^2\right)\right]-\frac{n\mu}{2}\log\left[2\pi
e\left(bP_1^*+1-\rho_2^2\right)\right],\label{eq:sum1}\eqa where
$P_1^*$ is defined in (\ref{eq:P1muL}) and (\ref{eq:P1muS}). For
the same reason, we have \bqa \mu h\left(X_2^n+Z_2^n\right)-
h\left(\sqrt{a}X_2^n+Z_1^n|N_1^n\right)\leq\frac{n\mu}{2}\log\left[2\pi
e\left(P_2^*+\sigma_2^2\right)\right]-\frac{n}{2}\log\left[2\pi
e\left(aP_2+1-\rho_1^2\right)\right],\label{eq:sum2}\eqa where
$P_2^*$ is defined in (\ref{eq:P2muL}) and (\ref{eq:P2muS}). From
(\ref{eq:exprate}), (\ref{eq:conditional1})-(\ref{eq:sum2}) we
obtain the rate constraint (\ref{eq:constraint1}).

On the other hand, we have \bqa
n(R_1+\eta_1R_2)&\leq&I\left(X_1^n;Y_1^n\right)+\eta_1I\left(X_2^n;Y_2^n\right)+n\epsilon\nn\\
&\leq&I\left(X_1^n;Y_1^n,X_2^n\right)+\eta_1I\left(X_2^n;Y_2^n\right)+n\epsilon\nn\\
&=&I\left(X_1^n;Y_1^n|X_2^n\right)+\eta_1I\left(X_2^n;Y_2^n\right)+n\epsilon\nn\\
&=&h\left(Y_1^n|X_2^n\right)-h\left(Y_1^n|X_1^n,X_2^n\right)+\eta_1h\left(Y_2^n\right)-\eta_1h\left(Y_2^n|X_2^n\right)+n\epsilon\nn\\
&=&h\left(X_1^n+Z_1^n\right)-\eta_1h\left(\sqrt{b}X_1^n+Z_2^n\right)-h\left(Z_1^n\right)+\eta_1h\left(Y_2^n\right)+n\epsilon\nn\\
&\leq&\frac{n}{2}\log\left(P_1^*+1\right)-\frac{n\eta_1}{2}\log\left(bP_1^*+1\right)+\frac{n\eta_1}{2}\log\left(1+bP_1+P_2\right)+n\epsilon,\label{eq:BCbound}\eqa
where the last step follows by Corollaries \ref{corollary:c1} and
\ref{corollary:c2}. We further have \bqa
P_1^*=\left\{\begin{array}{ll}
  P_1, & \quad \eta_1\leq\frac{1+bP_1}{b+bP_1} \\
  \frac{b\eta_1-1}{b-b\eta_1}, & \quad \frac{1+bP_1}{b+bP_1}\leq\eta_1\leq\frac{1}{b}\\
  0, & \quad \eta_1\geq\frac{1}{b}. \\
\end{array}\right.\label{eq:starP1}\eqa
Since the bounds in (\ref{eq:BCbound}) when $P_1^*=P_1$ and
$P_1^*=0$ are redundant, we have \bqa
R_1+\eta_1R_2&\leq&\frac{1}{2}\log\left(1+\frac{b\eta_1-1}{b-b\eta_1}\right)-\frac{\eta_1}{2}\log\left(1+\frac{b\eta_1-1}{1-\eta_1}\right)+\frac{\eta_1}{2}\log\left(1+bP_1+P_2\right),\eqa
for $\frac{1+bP_1}{b+bP_1}\leq\eta_1\leq\frac{1}{b}$, which is
(\ref{eq:constraint2}). We similarly obtain
(\ref{eq:constraint3}).

\subsection{Proof of Theorem \ref{thm:sumcapacity}}
By choosing \bqa
&&\hspace{-.35in}\sigma_1^2=\frac{1}{2b}\left\{b(aP_2+1)^2-a(bP_1+1)^2+1+\sqrt{\left[b(aP_2+1)^2-a(bP_1+1)^2+1\right]^2-4b(aP_2+1)^2}\right\}\label{eq:condition1}\\
&&\hspace{-.35in}\sigma_2^2=\frac{1}{2a}\left\{a(bP_1+1)^2-b(aP_2+1)^2+1+\sqrt{\left[a(bP_1+1)^2-b(aP_2+1)^2+1\right]^2-4a(bP_1\emph{}+1)^2}\right\}\\
&&\hspace{-.35in}\rho_1=\sqrt{1-a\sigma_2^2}\\
&&\hspace{-.35in}\rho_2=\sqrt{1-b\sigma_1^2},
\label{eq:condition4}\eqa the bound (\ref{eq:constraint1}) with
$\mu=1$ is \bqa
R_1+R_2\leq\frac{1}{2}\log\left(1+\frac{P_1}{1+aP_2}\right)+\frac{1}{2}\log\left(1+\frac{P_2}{1+bP_1}\right).\label{eq:upperlower}\eqa
By one can achieve equality in (\ref{eq:upperlower}) by treating
the interference as noise at both receivers.

In order that the choice of $\sigma_1^2$, $\sigma_2^2$, $\rho_1$
and $\rho_2$ be feasible, there must exist at least one pair
$(\sigma_1^2,\sigma_2^2)$ satisfying the following conditions:
\bqn \sigma_1^2\geq 0,\quad\sigma_2^2 \geq 0,\quad \rho_1\leq
1,\quad \rho_2\leq 1.\eqn Using
(\ref{eq:condition1})-(\ref{eq:condition4}), we thus require \bqa
\left[b(aP_2+1)^2-a(bP_1+1)^2+1\right]^2-4b(aP_2+1)^2&\geq&
0\label{eq:sigma1}\\
\left[a(bP_1+1)^2-b(aP_2+1)^2+1\right]^2-4a(bP_1\emph{}+1)^2&\geq&
0\label{eq:sigma2}\\
b(aP_2+1)^2-a(bP_1+1)^2+1&\geq&0\label{eq:rho1}\\
a(bP_1+1)^2-b(aP_2+1)^2+1&\geq&0.\label{eq:rho2}\eqa From
(\ref{eq:sigma1}) we have one of the following three
conditions\bqa \sqrt{b}\left(aP_2+1\right)-\sqrt{a}(bP_1+1)&\geq&
1,\label{eq:ineq1}\\
\sqrt{b}\left(aP_2+1\right)-\sqrt{a}(bP_1+1)&\leq&
-1,\label{eq:ineq2}\\
\sqrt{b}\left(aP_2+1\right)+\sqrt{a}(bP_1+1)&\leq&
1.\label{eq:ineq3}\eqa (\ref{eq:sigma2}) gives the same
constraints in (\ref{eq:ineq1})-(\ref{eq:ineq3}). Since
(\ref{eq:rho1}) and (\ref{eq:rho2}) exclude the possibilities
(\ref{eq:ineq1}) and (\ref{eq:ineq2}), this leaves
(\ref{eq:ineq3}) which is precisely (\ref{eq:power}) in Theorem
\ref{thm:sumcapacity}.

\subsection{Proof of Theorem \ref{thm:sumcapacity_mixed}}
The proof of (\ref{eq:constraint2}) requires only $0<b<1$.
Therefore (\ref{eq:constraint2}) is still valid when $a>1$.
Letting $\eta_1=1$ and $P_1^*=P_1$ in (\ref{eq:BCbound}) and
(\ref{eq:starP1}), we have the sum-rate capacity upper bound in
(\ref{eq:sumcapacity_mixed}). But (\ref{eq:sumcapacity_mixed}) is
 achievable if (\ref{eq:icmixed}) is true.  To verify this, we
let user $2$ communicate at
$R_2=\frac{1}{2}\log\left(1+\frac{P_2}{1+bP_1}\right)$. From
(\ref{eq:icmixed}), user $1$ can decode the message from user $2$
before decoding its own messages. Then we obtained
(\ref{eq:constraint_mixed}) and Theorem
\ref{thm:sumcapacity_mixed} is proved.

\section{Numerical examples}
A comparison of the outer bounds for a Gaussian IC is given in
Fig. \ref{fig:region}. Some part of the outer bound from Theorem
\ref{thm:region} overlaps with Kramer's outer bound due to
(\ref{eq:constraint2}) and (\ref{eq:constraint3}). Since this IC
has noisy interference, the proposed outer bound coincides with
the inner bound at the sum rate point.

The lower and upper bounds for the sum-rate capacity of the
symmetric IC($a=b,P_1=P_2$) are shown in Figs.
\ref{fig:sumVeryWeak}-\ref{fig:sumVeryStrongLog} for different
power levels. For all of these cases, the upper bounds are tight
up to point $A$. The bound in \cite[Theorem
3]{Etkin-etal:07IT_submission} approaches to the bound in Theorem
\ref{thm:region} when the power becomes large, but there is still
a gap. Fig. \ref{fig:sumVeryStrong} and \ref{fig:sumVeryStrongLog}
also provide a definitive answer to a question from \cite[Fig.
2]{Sason:04IT}: whether the sum-rate capacity of symmetric
Gaussian IC is a decreasing function of $a$, or there exists a
bump like the lower bound when $a$ varies from $0$ to $1$. In Fig.
\ref{fig:sumVeryStrong} and \ref{fig:sumVeryStrongLog}, our
proposed upper bound and Sason's inner bound explicitly show that
the sum capacity is not a monotone function of $a$ (this result
also follows by the bounds of \cite{Etkin-etal:07IT_submission}).

\section{Conclusions and extensions}

We derived an outer bound for the capacity region of Gaussian ICs
by a genie-aided method. From this outer bound, the sum-rate
capacities for ICs that satisfy  (\ref{eq:power}) or
(\ref{eq:constraint_mixed}) are obtained.

We discuss in the following some possible extensions of the
present work. One extension is already given in Remark 6 above.
Another extension is to generalize the sum-rate capacity for a
single noisy interference IC to that of parallel ICs, that occur
in, for instance, orthogonal frequency division multiplexing
(OFDM) systems. Finally, we note that the methods used in the
paper can also be applied to obtained bounds for multiple input
multiple output Gaussian ICs. We are currently developing such
bounds.

\vspace{.2in}

 \centerline{Acknowledgement} The work of X. Shang and B. Chen
was supported in part by the National Science Foundation under
Grants 0546491 and 0501534.

G. Kramer gratefully acknowledges the support of the Board of
Trustees of the University of Illinois Subaward no. 04-217 under
National Science Foundation Grant CCR-0325673 and the Army
Research Office under ARO Grant W911NF-06-1-0182.

\begin{figure}[htp]
\centerline{\leavevmode \epsfxsize=5in \epsfysize=3.4in
\epsfbox{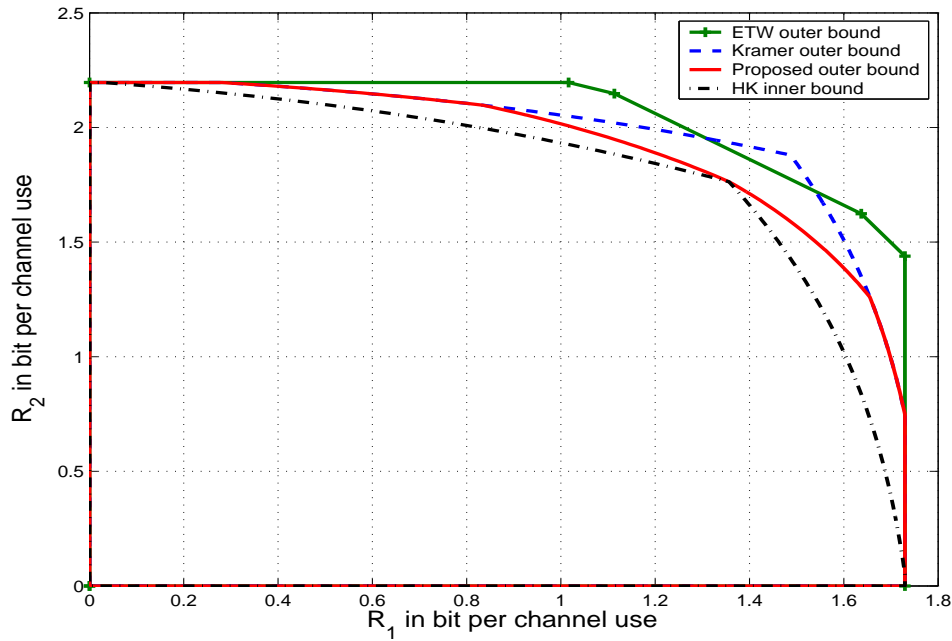}}\caption{Inner and outer bounds for the
capacity region of Gaussian ICs with
$a=0.09,b=0.04,P_1=10,P_2=20$. The ETW bound is by Etkin, Tse and
Wang in \protect\cite[Theorem 3]{Etkin-etal:07IT_submission}; the
Kramer bound is from \protect\cite[Theorem 2]{Kramer:04IT}; the HK
inner bound is based on \protect\cite{Han&Kobayashi:81IT} by Han
and Kobayashi.} \label{fig:region}\end{figure}

\begin{figure}[htp]
\centerline{\leavevmode \epsfxsize=5in \epsfysize=3.4in
\epsfbox{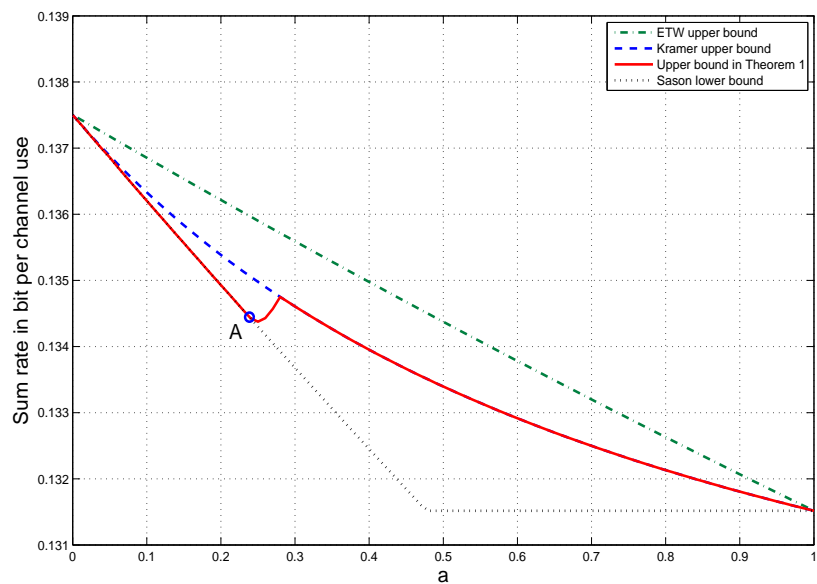}}\caption{Lower and upper bounds for the
sum-rate capacity of symmetric Gaussian ICs with
$a=b,P_1=P_2=0.1$. Sason's bound is an inner bound obtained from
Han and Kobayashi's bound by a special time sharing scheme
\protect\cite[Table I]{Sason:04IT}. The channel gain at point $A$
is $a=0.2385$.} \label{fig:sumVeryWeak}\end{figure}

\begin{figure}[htp]
\centerline{\leavevmode \epsfxsize=5in \epsfysize=3.4in
\epsfbox{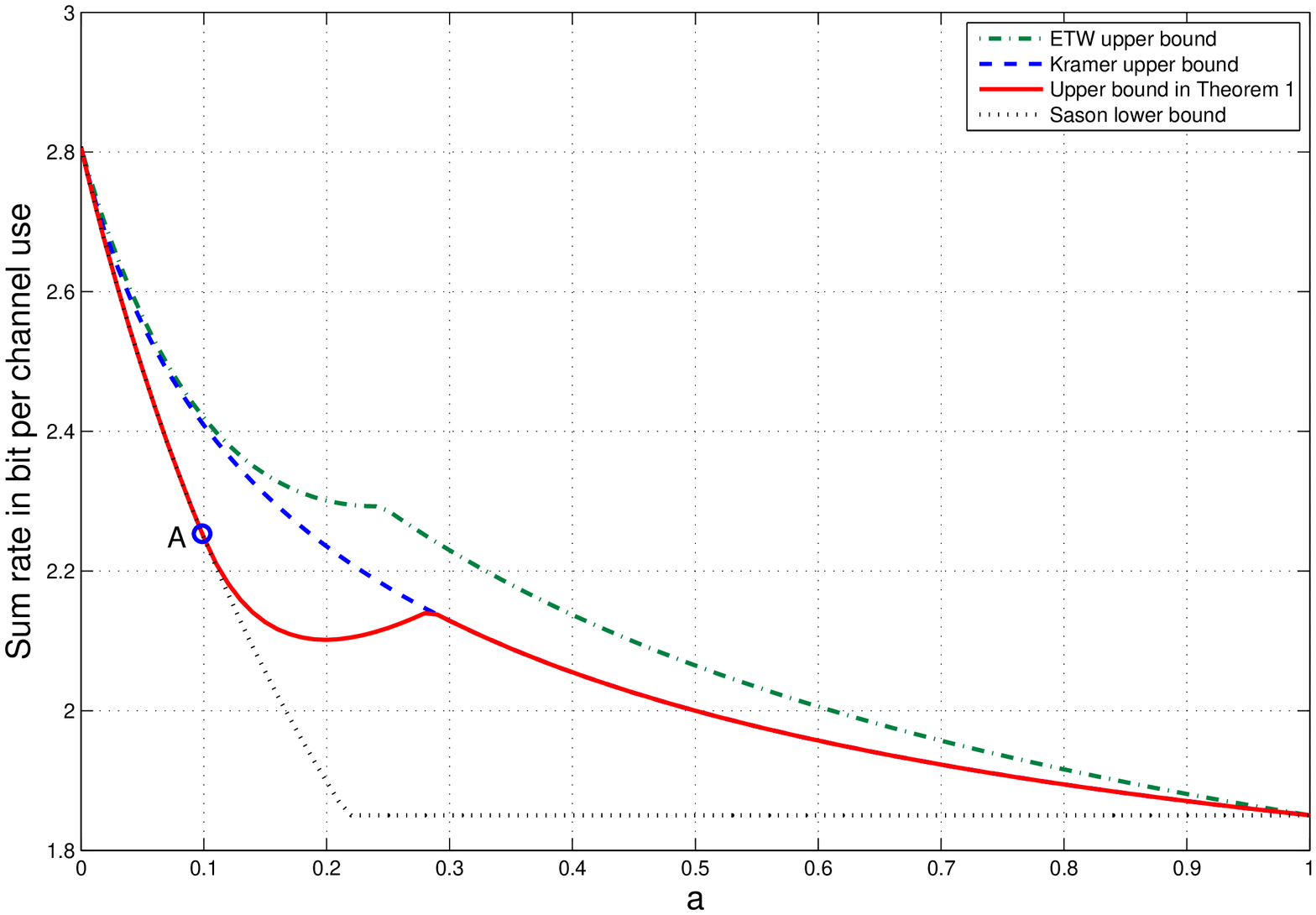}}\caption{Lower and upper bounds for the
sum-rate capacity of symmetric Gaussian ICs with $a=b,P_1=P_2=6$.
The channel gain at point $A$ is $a=0.0987$.}
\label{fig:sumMedian}\end{figure}

\begin{figure}[htp]
\centerline{\leavevmode \epsfxsize=5in \epsfysize=3.4in
\epsfbox{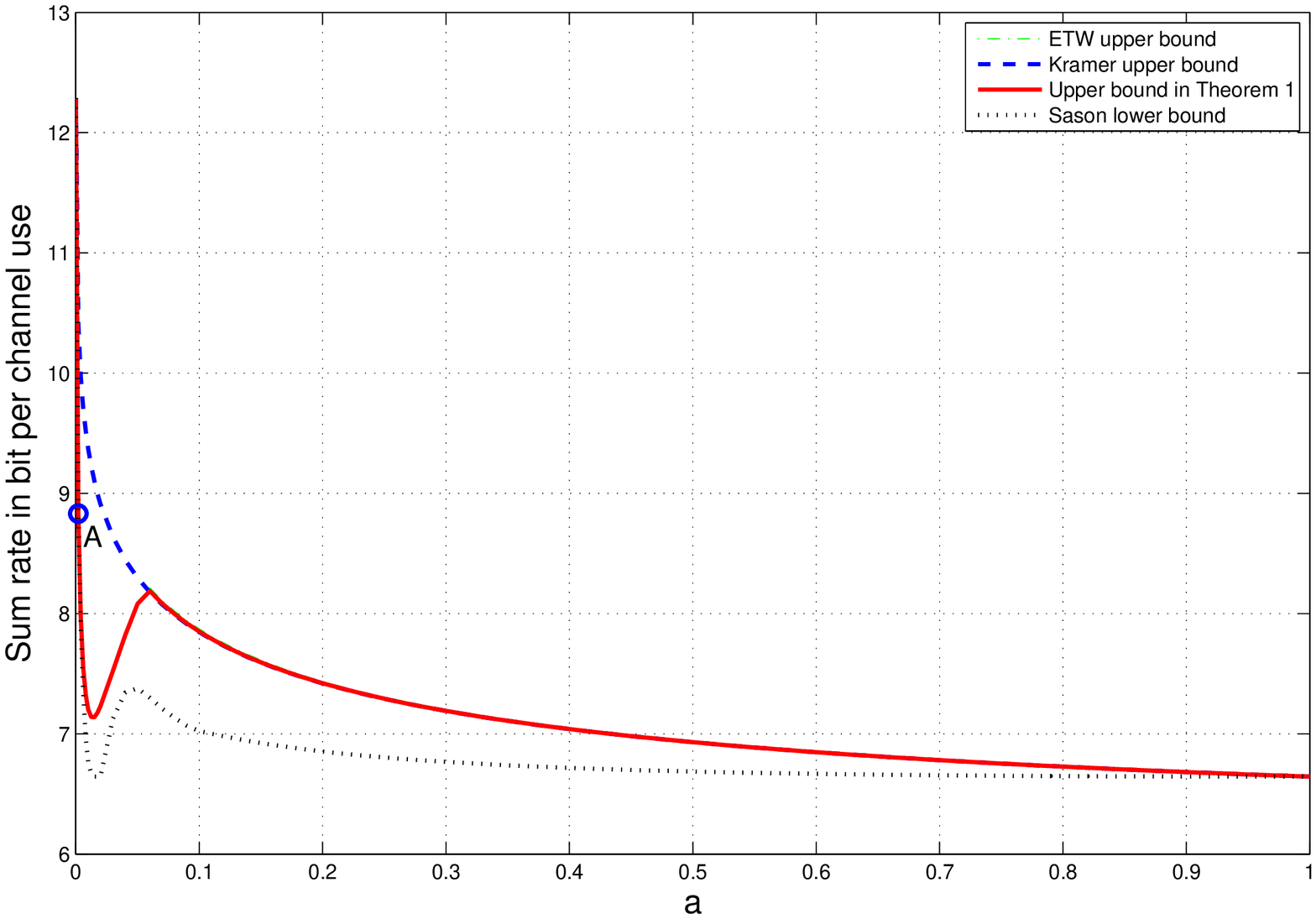}}\caption{Lower and upper bounds for the
sum-rate capacity of symmetric Gaussian ICs with
$a=b,P_1=P_2=5000$. The channel gain at point $A$ is
$a=0.002$.} \label{fig:sumVeryStrong}\end{figure}

\begin{figure}[htp]
\centerline{\leavevmode \epsfxsize=5in \epsfysize=3.4in
\epsfbox{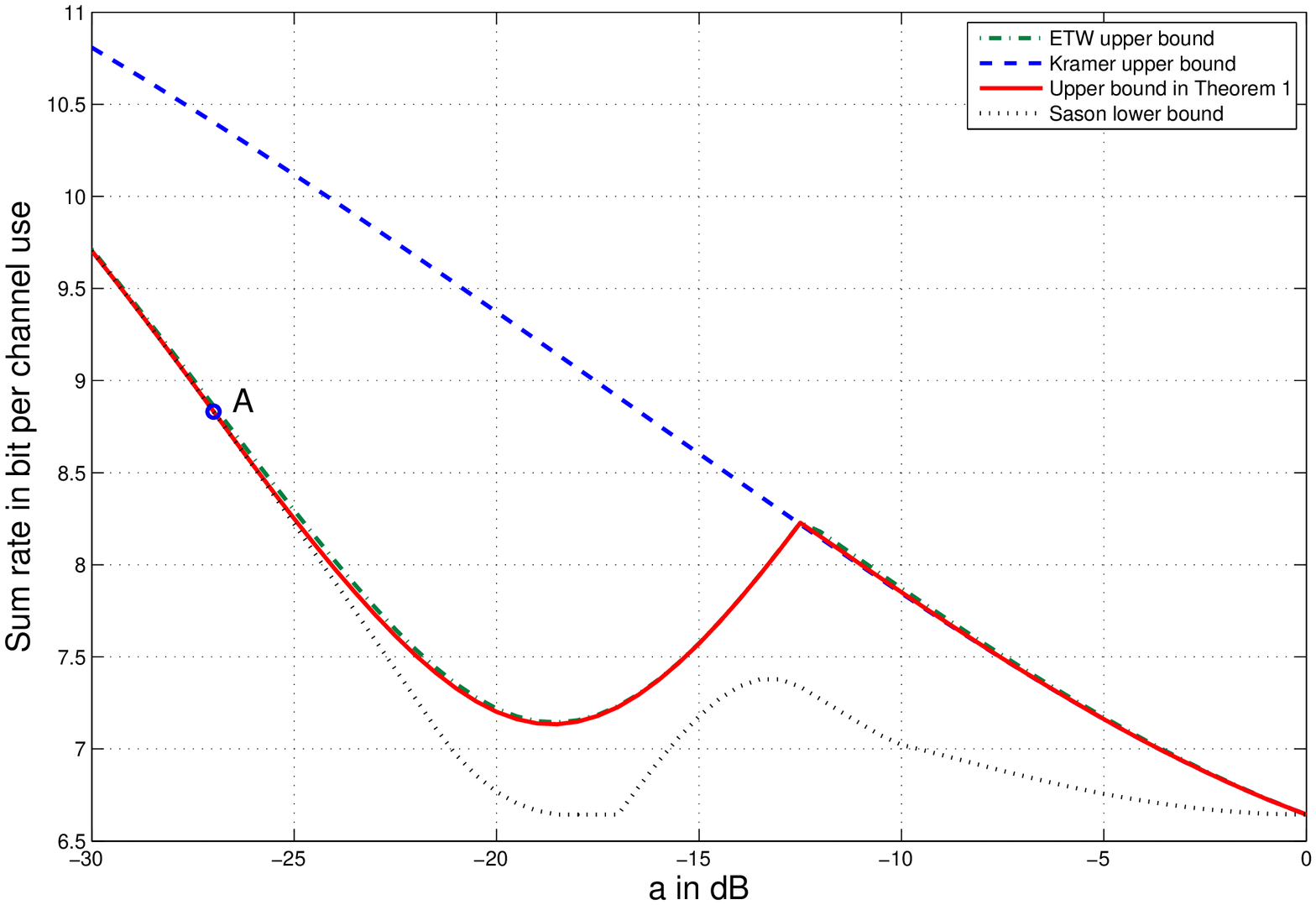}}\caption{Same as Fig.
\ref{fig:sumVeryStrong} with $a$ replaced by $10\log_{10}a$. The
channel gain at point $A$ is $a=-26.99$dB.}
\label{fig:sumVeryStrongLog}\end{figure}


\bibliography{\HOME/Journal,\HOME/Conf,\HOME/Misc,\HOME/Book}

\begin{thebibliography}{10}

\bibitem{Carleial:78IT}
A.B. Carleial,
\newblock ``{Interference Channels},''
\newblock {\em IEEE Trans. Inform. Theory}, vol. 24, pp. 60--70, Jan. 1978.

\bibitem{Carleial:75IT}
A.B. Carleial,
\newblock ``{A case where interference does not reduce capacity},''
\newblock {\em IEEE Trans. Inform. Theory}, vol. 21, pp. 569--570, Sep. 1975.

\bibitem{Sato:81IT}
H.~Sato,
\newblock ``{The capacity of the Gaussian interference channel under strong
  interference},''
\newblock {\em IEEE Trans. Inform. Theory}, vol. 27, pp. 786--788, Nov. 1981.

\bibitem{Han&Kobayashi:81IT}
T.S. Han and K.~Kobayashi,
\newblock ``{A new achievable rate region for the interference channel},''
\newblock {\em IEEE Trans. Inform. Theory}, vol. 27, pp. 49--60, Jan. 1981.

\bibitem{Costa:85IT}
M.H.M. Costa,
\newblock ``{On the Gaussian interference channel},''
\newblock {\em IEEE Trans. Inform. Theory}, vol. 31, pp. 607--615, Sept. 1985.

\bibitem{Chong-etal:06IT_submission}
H.F. Chong, M.~Motani, H.K. Garg, and H.E. Gamal,
\newblock ``{On the Han-Kobayashi Region for the interference channel},''
\newblock {\em submitted to the IEEE Trans. Inform. Theory}, 2006.

\bibitem{Kramer:06Zurich}
G.~Kramer,
\newblock ``{Review of rate regions for interference channels},''
\newblock in {\em {International Zurich Seminar}}, Feb. 2006.

\bibitem{Sato:77IT}
H.~Sato,
\newblock ``{Two-user communication channels},''
\newblock {\em IEEE Trans. Inform. Theory}, vol. 23, pp. 295--304, May 1977.

\bibitem{Carleial:83IT}
A.B. Carleial,
\newblock ``{Outer bounds on the capacity of interference channels},''
\newblock {\em IEEE Trans. Inform. Theory}, vol. 29, pp. 602--606, July 1983.

\bibitem{Kramer:04IT}
G.~Kramer,
\newblock ``{Outer bounds on the capacity of Gaussian interference channels},''
\newblock {\em IEEE Trans. on Inform. Theory}, vol. 50, pp. 581--586, Mar.
  2004.

\bibitem{Etkin-etal:07IT_submission}
R.~H. Etkin, D.~N.~C. Tse, and H.~Wang,
\newblock ``{Gaussian Interference Channel Capacity to Within One Bit},''
\newblock {\em submitted to the IEEE Trans. Inform. Theory}, 2007.

\bibitem{Telatar&Tse:07ISIT}
E.~Telatar and D.~Tse,
\newblock ``{Bounds on the capacity region of a class of interference
  channels},''
\newblock in {\em {Proc. IEEE International Symposium on Information Theory
  2007}}, Nice, France, Jun. 2007.

\bibitem{Liu&Viswanath:06IT}
T.~Liu and P.~Viswanath,
\newblock ``An extremal inequality motivated by multiterminal
  information-theoretic problems,''
\newblock {\em IEEE Trans. Inform. Theory}, vol. 53, no. 5, pp. 1839--1851, May
  2006.

\bibitem{Sato:78IT}
H.~Sato,
\newblock ``{On degraded Gaussian two-user channels},''
\newblock {\em IEEE Trans. Inform. Theory}, vol. 24, pp. 634--640, Sept. 1978.

\bibitem{Weingarten-etal:06IT}
H.~Weingarten, Y.~Steinberg, and S.~Shamai (Shitz),
\newblock ``{The Capacity Region of the Gaussian Multiple-Input Multiple-Output
  Broadcast Channel},''
\newblock {\em IEEE Trans. Inform. Theory}, vol. 52, no. 9, pp. 3936--3964,
  Sep. 2006.

\bibitem{Sason:04IT}
I.~Sason,
\newblock ``{On achievable rate regions for the Gaussian interference
  channels},''
\newblock {\em IEEE Trans. Inform. Theory}, vol. 50, pp. 1345--1356, June 2004.

\end{thebibliography}
\bibliographystyle{\HOME/IEEEbib}

\end{document}